\documentclass[letterpaper]{article} % DO NOT CHANGE THIS
\usepackage{aaai2027}  % DO NOT CHANGE THIS
\usepackage[hyphens]{url}  % DO NOT CHANGE THIS
\usepackage{graphicx} % DO NOT CHANGE THIS
\usepackage{natbib}  % DO NOT CHANGE THIS AND DO NOT ADD ANY OPTIONS TO IT
\usepackage{caption} % DO NOT CHANGE THIS AND DO NOT ADD ANY OPTIONS TO IT
\usepackage{algorithm}
\usepackage{algorithmic}

\usepackage{newfloat}
\usepackage{listings}
\DeclareCaptionStyle{ruled}{labelfont=normalfont,labelsep=colon,strut=off} % DO NOT CHANGE THIS
\floatstyle{ruled}
\newfloat{listing}{tb}{lst}{}
\floatname{listing}{Listing}

\usepackage{booktabs}
\usepackage{xspace}
\usepackage{algorithm}
\usepackage{algorithmic}
\usepackage{amsthm}
\usepackage{adjustbox}
\usepackage{tabularx}
\usepackage{multirow}
\usepackage{booktabs}
\usepackage{array}
\usepackage{amsmath}
\theoremstyle{definition}
\newtheorem{definition}{Definition}
\newcolumntype{Y}{>{\centering\arraybackslash}X}
\newtheorem{theorem}{Theorem}
\usepackage{amssymb}
\newcommand{\model}{DiGR\xspace}
\newtheorem{proposition}{Proposition}

\title{From Token Generation to Item Ranking: Direct Generative Recommendation with Semantic IDs}
\author{
    Yuanbo Zhao\textsuperscript{\rm 1}\equalcontrib,
    Ruochen Liu\textsuperscript{\rm 1}\equalcontrib,
    Senzhang Wang\textsuperscript{\rm 1}\corresponding,
    Jun Yin\textsuperscript{\rm 2},
    Yuxin Dong\textsuperscript{\rm 1},
    Huan Gong\textsuperscript{\rm 3},
    Hao Chen\textsuperscript{\rm 4},
    Shirui Pan\textsuperscript{\rm 5},
    Chengqi Zhang\textsuperscript{\rm 3}
}
\affiliations{
    \textsuperscript{\rm 1}Central South University\\
    \textsuperscript{\rm 2}Hong Kong Polytechnic University\\
    \textsuperscript{\rm 3}National University of Defense Technology\\
    \textsuperscript{\rm 4}City University of Macau\\
    \textsuperscript{\rm 5}Griffith University\\
    zhao\_yb@csu.edu.cn, ruochen@csu.edu.cn, szwang@csu.edu.cn, Junmay.yin@connect.polyu.hk, 8209230306@csu.edu.cn, gonghuan@nudt.edu.cn, sundaychenhao@gmail.com, s.pan@griffith.edu.au, chengqi.zhang@polyu.edu.hk
}

\begin{document}

\maketitle

\begin{abstract}
% A core objective in recommender systems is to accurately model the distribution of user preferences over items to enable personalized recommendations. Recently, driven by the strong generative capabilities of large language models (LLMs), LLM-based generative recommendation has become increasingly popular. However, we observe that existing methods inevitably introduce systematic bias when estimating item-level preference distributions. Specifically, autoregressive generation suffers from incomplete coverage due to beam search pruning, while parallel generation distorts probabilities by assuming token independence. We attribute this issue to a fundamental modeling mismatch: these methods approximate item-level distributions via token-level generation, which inherently induces approximation errors. Through both theoretical analysis and empirical validation, we demonstrate that token-level generation cannot faithfully substitute item-level generation, leading to biased item distributions. To address this, we propose \textbf{Sim}ply \textbf{G}enerative \textbf{R}ecommendation (\textbf{SimGR}), a framework that directly models item-level preference distributions in a shared latent space and ranks items by similarity, thereby aligning the modeling objective with recommendation and mitigating distributional distortion. Extensive experiments across multiple datasets and LLM backbones show that SimGR consistently outperforms existing generative recommenders. 

Generative recommendation formulates item recommendation as a token-level generation task, where Semantic IDs (SIDs) represents each item as a sequence of discrete tokens. However, recommendation ultimately requires item-level rankings, whereas SID-based methods derive them by decoding token-level outputs. We term these outputs the token interface; together, the interface and decoder form a token-mediated pipeline. We establish a theoretical dichotomy: if the interface is ranking-insufficient, no decoder based solely on it can guarantee exact ranking recovery; if it is ranking-sufficient, exact decoding is output-equivalent to item-level scoring. Thus, for item ranking, token-level generation either loses essential ranking information or provides no additional ranking expressiveness beyond direct item-level scoring.
 Based on this insight, we propose \textbf{Di}rect \textbf{G}enerative \textbf{R}ecommendation (\model), a framework that directly models item-level preferences while preserving the semantic structure of SID. Instead of treating SID tokens as generation targets, \model uses them as item representations and learns user-item matching through a unified item-level scoring function. Extensive experiments on multiple real-world datasets with LLM backbones of different scales demonstrate that \model consistently outperforms existing generative recommenders as well as ID-based methods.
\end{abstract}

\section{Introduction}

Generative recommendation has fundamentally reshaped recommender systems by formulating item recommendation as a token-level generation task via Large Language Models (LLMs). In this paradigm, Semantic ID (SID) has become an increasingly adopted item representation~\cite{rajput2023recommender,zheng2024adapting}. 
Through vector quantization techniques such as RQ-VAE~\cite{NIPS2017_7a98af17,lee2022autoregressive}, SIDs represent each item as a sequence of discrete semantic tokens. Their shared semantic structure improves generalization to new and long-tail items and previously unseen item transitions, while yielding more stable representations under catalog shifts~\cite{singh2024better,ding2026does,zheng2025enhancing}. These advantages have made SIDs a widely adopted item representation in generative recommendation, where recommendation is formulated as a token-level generation task.
However, the ultimate goal of recommendation remains item-level preference ranking, while SID-based generative recommendation requires recovering item-level preferences from multi-step token predictions~\cite{hou2024large,zhang2025recommendation,singh2024better}. Therefore, a fundamental question is whether token-level prediction can effectively preserve and recover the ranking information required for item recommendation.

\begin{figure}[t]
    \centering
    \includegraphics[width=1.0\linewidth]{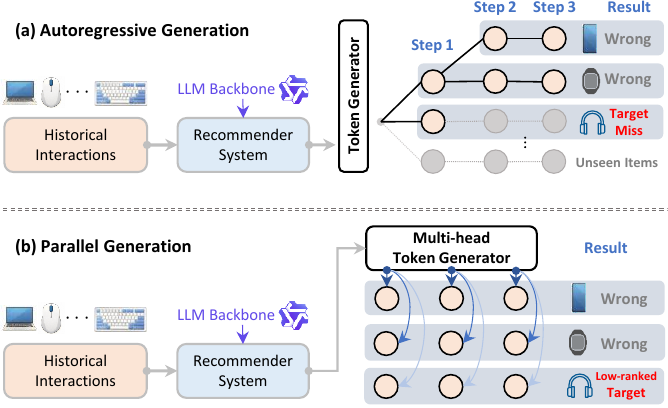}
    \caption{Illustration of ranking-information loss in SID-based generation. (a) Autoregressive decoding explores only a small subset of valid SID sequences constrained by beam search, omitting many candidate items. (b) Parallel decoding multiplies token-level probabilities across positions, which can undesirably lower the score of the target item.}
    \label{fig:limitation}
\end{figure}

As illustrated in Fig.~\ref{fig:limitation}, existing approaches generally adopt two types of generation paradigms: autoregressive generation and parallel generation~\cite{li2025survey}. \textbf{Autoregressive approaches} (Fig.~\ref{fig:limitation}a) follow the standard next-token prediction pipeline, predicting each token sequentially and finally applying beam search over a prefix-tree-constrained search space to retain high-probability paths for recovering item rankings~\cite{deng2025onerec,zheng2024adapting}. In contrast, \textbf{parallel approaches} (Fig.~\ref{fig:limitation}b) employ multiple prediction heads to simultaneously estimate token distributions for all positions in a single forward pass, and combine the predicted distributions across positions to derive item scores, thereby avoiding the overhead of sequential decoding~\cite{hou2025generating,shi2025llada}. Despite their different generation mechanisms, both approaches share a fundamental assumption: item-level preference ranking can be indirectly recovered from multiple token-level predictions.

However, such indirect recovery inevitably introduces systematic biases. Specifically, (i) autoregressive decoding requires searching over an exponentially growing token space and can only retain a limited number of high-probability paths, making it difficult to fully preserve the item-level preference distribution~\cite{lin2025order}; (ii) parallel decoding methods, while covering all token positions simultaneously, rely on approximation assumptions when aggregating token-level predictions, such as combining marginal probabilities across positions to derive item scores, which may distort the underlying item-level probability structure. We further prove that these limitations are not incidental flaws of particular decoding strategies. Instead, they arise from a fundamental constraint on recovering item rankings from token-level predictions. If multi-step token predictions preserve complete item-level ranking information, exact recovery from them is mathematically equivalent, at the output level, to directly scoring each item. Otherwise, two preference distributions can produce identical token prediction outputs while inducing different item rankings. No decoder can correctly recover both rankings from the same outputs. Thus, multi-step token prediction either provides no additional ranking expressiveness beyond direct item scoring or is insufficient for exact item-level ranking recovery.

This dichotomy reveals that the representational and predictive roles of SIDs are not inherently coupled. It is therefore possible to retain SIDs as compositional item representations while eliminating token-level SID decoding from the prediction interface. Based on this separation, we propose \textbf{Di}rect \textbf{G}enerative \textbf{Recommender} (\model), which retains SID-based item composition but replaces token generation with direct item-level scoring. Specifically, an LLM encodes each user's interaction history into a latent representation. In parallel, \model aggregates the SID token embeddings of each item into a compositional item representation. It then computes preference scores by comparing the two representations in a shared semantic space and ranks the entire item set. Using the same item-level scoring function during training and inference eliminates the distribution shift introduced by intermediate token-level recovery. Our contributions are summarized as follows:

\begin{itemize}
\item We identify an inherent theoretical bottleneck of token-level decoding: multi-step token prediction is either insufficient for accurate item ranking or equivalent to direct item-level scoring.
\item Motivated by this analysis, we propose \model, which preserves the semantic information of SID while directly parameterizing user preference at the item level, thereby avoiding the systematic bias introduced by multi-step token decoding.
\item Extensive experiments on multiple real-world datasets and LLM backbones of different scales demonstrate that \model consistently outperforms existing generative recommendation methods.
\end{itemize}

\section{Related Work}
\subsection{LLM-enhanced Discriminative Recommendation}
 
LLMs have been increasingly incorporated into discriminative recommenders,
where items remain atomic candidates ranked by user--item compatibility
scores. Rather than changing this inference interface, existing methods
mainly use LLMs to enrich recommendation representations. AlphaRec combines
language-derived item representations with collaborative filtering
components~\cite{sheng2025language}, while RLMRec constructs semantic user
and item profiles and aligns them with collaborative representations
through cross-view learning~\cite{ren2024representation}. These studies
demonstrate the value of LLM-derived knowledge while retaining conventional
item-level scoring.

\subsection{Generative Recommendation}

Generative recommendation instead formulates next-item prediction as sequence
generation, with existing methods differing primarily in item representation
and catalog grounding. Text-grounded approaches such as BIGRec generate
item-related language and subsequently map it to actual items~\cite{10.1145/3716393}.
A prominent alternative represents each item with a sequence of semantic ID
tokens. TIGER establishes the autoregressive SID generation paradigm, while
LC-Rec further integrates semantic and collaborative information into learned
identifiers~\cite{rajput2023recommender,zheng2024adapting}. Recent methods
reconsider the decoding process: LLaDA-Rec employs discrete diffusion for
bidirectional SID generation, whereas RPG predicts long unordered SIDs with
graph-guided decoding~\cite{shi2025llada,hou2025generating}. Despite their
different representations and generation orders, these methods obtain item
predictions through token generation and subsequent item identification. In
contrast, our model family preserves SID compositionality while performing
prediction directly at the item level.

\section{Analysis of Token-Mediated Item Ranking}
\label{sec:motivation}

This section analyzes whether token-mediated inference faithfully
preserves the item ranking induced by a generative recommender. We
separate this question into two requirements: whether the token
interface contains sufficient information to identify the item ranking,
and whether the decoder recovers that ranking exactly from the
interface. This distinction allows us to analyze autoregressive and
parallel generation within a unified framework.

\subsection{A Unified View: Preserving Model-Induced Item Rankings}

In SID-based generative recommenders, SIDs represent each item $i\in\mathcal{I}$ as a
sequence of discrete semantic tokens,
\[
\mathbf{s}_i=(s_{i,1},\ldots,s_{i,L}).
\]
Although they predict tokens, the recommendation task ultimately require a item-level preference ranking. They must therefore expose the item preferences
learned by the model through a token-level interface and then recover an
item-level ranking from that interface via a decoder.

Our analysis focuses on whether this token-mediated inference process faithfully preserves the item preferences learned by the model. For a fixed user context \(x\), let
\(p(\cdot\mid x)\) denote the model-induced item-level preference distribution,
where \(p(i\mid x)\) is the preference assigned to item \(i\). The corresponding
exact model ranking is obtained by sorting all items according to these
preferences:
\[
\mathcal{R}(p)
=
\operatorname{argsort}^{\downarrow}_{i\in\mathcal{I}}p(i\mid x).
\]

For an arbitrary SID generation paradigm $G$, let $\Phi_G(p)$ denote the
token-level information through which the model-induced preferences $p$ are
made available before item-level decoding. For example, $\Phi_G(p)$ may
consist of next-token conditional probabilities in autoregressive generation
or position-wise token probabilities in parallel generation. A decoder $D_G$
then converts this information into the final item ranking prediction:
\[
\widehat{\mathcal{R}}_G(p)
=
D_G\!\left(\Phi_G(p)\right).
\]
This decomposition separates an information-theoretic question from an
algorithmic one: whether \(\Phi_G(p)\) contains sufficient information to
uniquely determine \(\mathcal{R}(p)\), and, when it does, whether the
particular decoder \(D_G\) recovers \(\mathcal{R}(p)\) exactly from that
information.

We call $\Phi_G$ \emph{globally ranking-sufficient} if, for any two possible
model-induced preference distributions $p$ and $\widetilde p$,
\[
\Phi_G(p)=\Phi_G(\widetilde p)
\quad\Longrightarrow\quad
\mathcal{R}(p)=\mathcal{R}(\widetilde p).
\]
This condition requires each output of the token interface to identify a unique item
ranking. It therefore determines whether exact ranking recovery is possible
from the information exposed by the interface. The resulting alternatives are
formalized below.

\begin{proposition}
\label{prop:token_ranking_dichotomy}

For a token interface $\Phi_{\mathcal{G}}$, the following two
statements are equivalent:

\begin{enumerate}
    \item $\Phi_{\mathcal{G}}$ is globally ranking-sufficient.
    
    \item There exists an item-level score mapping
    \[
    \mathcal{S}_{\mathcal{G}}:
    \operatorname{Range}(\Phi_{\mathcal{G}})
    \rightarrow
    \mathbb{R}^{|\mathcal{I}|}
    \]
    such that, for every model-induced item distribution $p$,
    \[
    \operatorname{argsort}^{\downarrow}_{i\in\mathcal{I}}
    \mathcal{S}_{\mathcal{G}}
    \bigl(
        \Phi_{\mathcal{G}}(p)
    \bigr)_i
    =
    R(p).
    \]
\end{enumerate}

If these equivalent conditions do not hold, there exist two
distributions $p$ and $\widetilde{p}$ satisfying
\[
\Phi_{\mathcal{G}}(p)
=
\Phi_{\mathcal{G}}(\widetilde{p}),
\qquad
R(p)
\neq
R(\widetilde{p}),
\]
and no decoder or item-level score mapping based only on
$\Phi_{\mathcal{G}}$ can recover the exact ranking for both.
\end{proposition}

\begin{proof}
Suppose first that $\Phi_{\mathcal{G}}$ is globally
ranking-sufficient. Then the ranking is constant over every fiber of
$\Phi_{\mathcal{G}}$. Hence, for each
$y\in\operatorname{Range}(\Phi_{\mathcal{G}})$, we may define
$\overline{R}(y)=R(p)$ using any $p$ such that
$\Phi_{\mathcal{G}}(p)=y$. This definition is well-defined. Assigning
distinct decreasing scores according to $\overline{R}(y)$ yields a mapping
$\mathcal{S}_{\mathcal{G}}$ whose descending argsort is exactly
$\overline{R}(y)$, and therefore recovers $R(p)$.

Conversely, if such a score mapping exists, then
$\Phi_{\mathcal{G}}(p)=\Phi_{\mathcal{G}}(\widetilde{p})$ implies that
the two distributions receive the same score vector and hence the same
ranking. Thus, $\Phi_{\mathcal{G}}$ is globally ranking-sufficient.

Finally, if the conditions fail, global ranking sufficiency fails by
definition, so there exist $p$ and $\widetilde{p}$ with the same interface
output but different rankings. Any decoder or score mapping depending only
on that output must return the same result for both and cannot recover both
rankings exactly.
\end{proof}

The detailed formalized proof is provided in Appendix A.
% \begin{proof}
% If $\Phi_G$ is not globally ranking-sufficient, the first case follows
% directly from the negation of the definition. Since every decoder receives the
% same input for $p$ and $\widetilde p$, it must return the same ranking for
% both, which cannot equal both $\mathcal{R}(p)$ and
% $\mathcal{R}(\widetilde p)$.

% If $\Phi_G$ is globally ranking-sufficient, all preference distributions
% mapped to the same token-interface output share one exact model ranking.
% Defining $D_G^{\star}$ to return this unique ranking yields
% $D_G^{\star}(\Phi_G(p))=\mathcal{R}(p)$ for every $p$. Assigning item scores
% according to their positions in this recovered ranking produces an item-level
% score vector with the same ordering. The two cases are mutually exclusive and
% exhaustive.
% \end{proof}

% Proposition~\ref{prop:token_ranking_dichotomy} characterizes the fundamental
% limit of a token interface, but a ranking-sufficient interface does not by
% itself guarantee that a particular decoder is exact. 
% End-to-end ranking fidelity therefore requires both
% \begin{enumerate}
%     \item a globally ranking-sufficient token interface; and
%     \item a decoder guaranteed to recover the exact full-catalog ranking from
%     the information exposed by that interface.
% \end{enumerate}

The central implication of Proposition~\ref{prop:token_ranking_dichotomy} goes beyond the distinction between ranking-sufficient and ranking-insufficient interfaces. When a ranking-sufficient interface is coupled with an exact decoder, the resulting token-mediated inference is output-equivalent to direct item-level scoring: both produce exactly the same model-induced item ranking. Thus, from the perspective of ranking preservation, exact token decoding constitutes an indirect realization of item-level scoring rather than an indispensable inference mechanism. This equivalence motivates our subsequent development of an item-atomic paradigm that preserves SID compositionality while performing ranking directly over items. It is restricted to the final ranking and does not imply equivalence in parameterization, training or inference.

% Here, \emph{global search} refers to any inference procedure that is guaranteed
% to return the exact full-catalog model ranking, whether by exhaustive scoring
% or by an equivalent exact structured search. If the interface is
% ranking-insufficient, no decoder can repair the missing information. If the
% interface is sufficient but the decoder is only approximate, the available
% ranking information may still not be recovered.

\subsection{Existing Paradigms under the Unified View}

% Autoregressive and parallel generation instantiate the two failure modes
% identified above. Autoregressive modeling can preserve a complete
% model-induced item distribution, but practical beam search may fail to recover
% its global ranking. Parallel generation avoids sequential search, but its
% position-wise token interface may be insufficient to represent the model's
% item ordering in the first place.

We next apply this unified view to the two dominant SID generation paradigms:
autoregressive and parallel generation. For each paradigm, we examine how
model-induced item preferences are exposed through its token interface and
subsequently converted into the final item ranking.

\paragraph{Autoregressive generation: an inexact decoder.}
Autoregressive generation sequentially generates SIDs. Therefore, its token interface consists of the conditional probabilities over all valid prefixes.
For every item $i$,
\[
p(i\mid x)
=
\prod_{l=1}^{L}
p\!\left(
s_{i,l}\mid \mathbf{s}_{i,<l},x
\right).
\]
Consequently, the complete autoregressive interface is globally
ranking-sufficient: evaluating every valid SID recovers the exact model
ranking $\mathcal{R}(p)$. Autoregressive factorization itself therefore does
not necessarily distort the model's learned item ordering.

The remaining challenge lies in the decoder. A straightforward way to recover \(\mathcal{R}(p)\) is to evaluate the complete-sequence score of every legal SID sequence and rank all items accordingly. However, traversing a catalog-scale number of SID paths incurs prohibitive inference latency. Practical autoregressive recommenders therefore use beam search, which retains only the \(B\) highest-scoring prefixes at each step. Ranking partial prefixes, however, is not equivalent to ranking their descendant items: a pruned prefix may still contain an item belonging to the global top-\(k\). Consequently, for some model-induced preference distributions \(p\), \[ D_{\mathrm{beam}}\!\left(\Phi_{\mathrm{AR}}(p)\right) \neq \mathcal{R}(p). \] Beam search thus reduces inference cost at the expense of an exact recovery guarantee.

% We empirically examine this discrepancy using LC-Rec, a representative
% autoregressive generative recommender. We first train one fixed model and then
% perform inference with different beam sizes. The top-$k$ list obtained with a
% sufficiently large beam is used as a practical approximation to the model's
% near-global ranking. Table~\ref{tab:pre} reports the ranking overlap,
% recommendation accuracy, decoding coverage, and inference time. A narrow beam
% changes the ranking returned by the same model, while increasing the beam
% improves agreement with near-global search at a clear computational cost.
% These results directly demonstrate the trade-off between decoding efficiency
% and fidelity to the model-induced item ranking.

\paragraph{Parallel generation: an insufficient token interface.}
Parallel generation predicts all SID positions simultaneously. Under ideal
position-wise estimation, its token interface consists of the marginal
preferences
\[
m_l(t\mid x)
=
\sum_{i:s_{i,l}=t}p(i\mid x),
\qquad
\Phi_{\mathrm{PAR}}(p)
=
\left\{m_l(t\mid x)\right\}_{l,t}.
\]
A parallel decoder reconstructs the score of item \(i\) as
\[
q_{\mathrm{PAR}}(i\mid x)
\propto
\prod_{l=1}^{L}m_l(s_{i,l}\mid x).
\]
Although this score can be evaluated for every valid item, the following
theorem shows that the position-wise marginals do not generally preserve the
exact model ranking.

\begin{theorem}
\label{thm:parallel_insufficiency}
Let \(\mathcal{V}_l\) denote the token vocabulary at SID position \(l\).
If
\[
|\mathcal{I}|
>
1+\sum_{l=1}^{L}\left(|\mathcal{V}_l|-1\right),
\]
then \(\Phi_{\mathrm{PAR}}\) is not globally ranking-sufficient. Specifically,
there exist two model-induced preference distributions \(p^{+}\) and \(p^{-}\)
such that
\[
\Phi_{\mathrm{PAR}}(p^{+})
=
\Phi_{\mathrm{PAR}}(p^{-}),
\qquad
\mathcal{R}(p^{+})
\neq
\mathcal{R}(p^{-}).
\]
Consequently, no decoder based only on the position-wise token marginals can
recover the exact model ranking for every preference distribution. 

\end{theorem} 

The detailed proof of Theorem~\ref{thm:parallel_insufficiency} is in Appendix B.

Theorem~\ref{thm:parallel_insufficiency} shows that different model-induced
item rankings can produce identical position-wise token marginals. Any decoder
based only on these marginals must therefore return the same ranking for both
and fail on at least one of them. Exhaustively scoring the full candidate set
removes search approximation, but cannot recover the ranking information
already lost by \(\Phi_{\mathrm{PAR}}\).

\subsection{Design Motivation}

The preceding analysis identifies two sources of ranking mismatch in
token-mediated inference. A ranking-insufficient interface cannot
uniquely determine the model-induced item ranking, whereas a
ranking-sufficient interface still requires an exact decoder to recover
that ranking. When both requirements are satisfied, the resulting
decision admits an output-equivalent item-level score representation.
This motivates us to parameterize complete-item scores directly while
retaining SIDs as compositional item representations. The next section
formalizes the resulting model family.

\section{Methodology}
\subsection{Item-Atomic SID-Based Model Family}

We first formalize the class of SID-based recommenders motivated by the
preceding analysis. The definition is independent of any particular
context encoder, SID composition operator, or item-scoring function.

\begin{definition}
\label{def:simgr_family}

Let \(\mathcal{I}\) be a finite item set, where each item
\(i\in\mathcal{I}\) is uniquely identified by an SID
\(\mathbf{s}_i=(s_{i,1},\ldots,s_{i,L})\). An item-atomic SID-based
recommender is defined as
\begin{align}
\mathbf{v}_i
&=
\Psi_{\phi}
\left(
\mathbf{e}(s_{i,1}),\ldots,\mathbf{e}(s_{i,L})
\right),
\qquad i\in\mathcal{I},
\label{eq:aggregation}
\\
z_{\Theta}(x,i)
&=
g_{\eta}\left(\mathbf{f_\theta}(x),\mathbf{v}_i\right),
\qquad
\Theta=(\theta,\phi,\eta),
\label{eq:item_score}
\end{align}
where \(f_{\theta}\) encodes the user context, \(\mathbf{e}\) is the SID-token
embedding lookup, \(\Psi_{\phi}\) composes the complete SID into an item
representation, and \(g_{\eta}\) assigns a preference score to each complete
item.

The resulting item-level preference distribution and ranking are formalized as:
\begin{align}
p_{\Theta}(i\mid x)
&=
\frac{
\exp\left(z_{\Theta}(x,i)\right)
}{
\sum_{j\in\mathcal{I}}
\exp\left(z_{\Theta}(x,j)\right)
},
\label{eq:item_distribution}
\\
\mathcal{R}_{\Theta}(x)
&=
\operatorname{argsort}^{\downarrow}_{i\in\mathcal{I}}
z_{\Theta}(x,i).
\label{eq:item_ranking}
\end{align}
\end{definition}

% Definition~\ref{def:simgr_family} separates the semantic role of SIDs from the
% ranking decision. SIDs are used to construct complete-item representations,
% whereas both preference modeling and ranking operate directly on the resulting
% item scores. Different choices of \(f_{\theta}\), \(\Psi_{\phi}\), and
% \(g_{\eta}\) instantiate different members of this family.

Definition~1 separates the representation role of SIDs from the final
ranking decision. The SID embeddings are composed into complete-item
representations by $\Psi_\phi$, while preference modeling and ranking
operate on the resulting item scores. Different choices of
$f_\theta$, $\Psi_\phi$, and $g_\eta$ define different members of this
family. We next present \model as a parsimonious instantiation.

\subsection{\model: A Parsimonious Instantiation}

The model family defined above leaves the context encoder
$f_{\theta}$, SID composition operator $\Psi_{\phi}$, and item-level
scorer $g_{\eta}$ unspecified. \model instantiates them with an LLM
encoder, parameter-free mean pooling, and temperature-scaled
dot-product scoring, respectively.

\paragraph{Context encoding.}
Given an interaction context $x$, the LLM backbone produces a
$d$-dimensional representation
$\mathbf{h}_x=f_{\theta}(x)$, which is derived from the hidden of the last layer in \(f_\theta\) before candidate-item scoring.

\paragraph{SID composition.}
For an item composed of SID
$\mathbf{s}_i=(s_{i,1},\ldots,s_{i,L})$, \model constructs its
representation as
\begin{equation}
    \mathbf{v}_i
    =
    \frac{1}{L}
    \sum_{\ell=1}^{L}
    \mathbf{e}_{\ell}(s_{i,\ell}),
    \label{eq:digr_item_representation}
\end{equation}
where $\mathbf{e}_{\ell}(s_{i,\ell})$ is the position-specific
SID-token. This operator is instantiated as look-up in the vocabulary of the LLM backbone. The embeddings are shared across items, and mean pooling introduces no additional
composition parameters and no item-specific information.

\paragraph{Item-level scoring, optimization, and ranking.}
Let
$\mathbf{V}=[\mathbf{v}_1,\ldots,\mathbf{v}_{|\mathcal{I}|}]^{\top}$
denote the matrix of catalog representations. Given a context
representation $\mathbf{h}_x$, \model computes the catalog-level score
vector as
\begin{equation}
    \mathbf{z}_{\theta}(x)
    =
    \frac{\mathbf{V}\mathbf{h}_x}{\tau},
    \qquad
    p_{\theta}(i\mid x)
    =
    \frac{\exp(z_{\theta}(x,i))}
    {\sum_{j\in\mathcal{I}}\exp(z_{\theta}(x,j))},
    \label{eq:digr_scoring}
\end{equation}
where $\tau>0$ controls the concentration of the item distribution.
Given the ground-truth next item $i^{+}$, the model is optimized with
the item-level negative log-likelihood
\begin{equation}
    \mathcal{L}_{\mathrm{item}}
    =
    -\log p_{\theta}(i^{+}\mid x).
    \label{eq:digr_loss}
\end{equation}
Once the model parameters are fixed, $\mathbf{V}$ can be precomputed,
and top-$K$ recommendation is obtained directly from the largest
entries of the same score vector $\mathbf{z}_{\theta}(x)$. Consequently,
both learning and retrieval compare complete items under an identical
scoring function; no intermediate token-level objective or decoding
rule changes the ranking criterion between training and inference.

\subsection{Discussion}

\paragraph{Inductive capability.}
The proposed family removes SID decoding while retaining SID
compositionality. Because each item representation is constructed
solely from shared SID-token embeddings, a new item can be scored once
its SID is available and added to the candidate catalog, without
learning item-specific parameters. The family therefore supports
inductive inference at the representation level, although its
accuracy also depends on SID quality and token support.

\paragraph{Generative nature and alternative interfaces.}
The proposed family is generative because it directly parameterizes the
conditional next-item distribution ($p_{\Theta}(i\mid x)$ and learns
this distribution by maximizing next-item likelihood. Its
generative nature therefore lies in probabilistic modeling of the
recommendation outcome, rather than in token-level decoding. Existing
methods such as BIGRec~\cite{10.1145/3716393},
IDGenRec~\cite{tan2024idgenrec}, and SETRec~\cite{lin2025order} generate
intermediate token, textual, or set representations before recovering
catalog items. In contrast, our family treats each complete item as an
atomic generation target and uses its SID only to construct the item
representation. It thus changes the generation space from token
sequences to complete items, rather than abandoning the generative
formulation.

\paragraph{Flexibility and scalability.}
\model is a parsimonious instantiation of the general item-level family. Its mean SID composition operator and dot-product scorer can be replaced by more expressive alternatives, provided that the resulting model scores complete items and defines an item-level conditional distribution used consistently during training and inference. To preserve inductive inference, item representations must be constructed from shared SID-level parameters without introducing item-specific parameters. Exact ranking requires scoring the full catalog. For very large catalogs, approximate item-level retrieval can reduce inference cost. Any resulting error then arises from search approximation over items, rather than information loss through a token-generation interface.

\section{Experiment}
\label{sec:exp}

We evaluate \model from multiple perspectives by answering the following research questions:

\noindent\textbf{RQ1: Overall effectiveness.}
How does \model compare with representative ID-based sequential recommenders
and generative recommenders in next-item ranking?

\noindent\textbf{RQ2: Ranking agreement and inference trade-offs.}
How do direct and token-mediated inference compare in ranking agreement, accuracy, and efficiency?

\noindent\textbf{RQ3: Backbone configurations.}
Does the performance advantage of \model persist across the LLM
backbones?

\noindent\textbf{RQ4: SID composition.}
How do different SID composition operators affect recommendation performance?

\noindent\textbf{RQ5: Exposure diversity.}
How broadly does \model distribute its recommendations across item categories?

\subsection{Experimental Setup}

\subsubsection{Datasets}

We evaluate \model on three product categories from Amazon Reviews
2018~\cite{ni2019justifying}: Musical Instruments, Arts, Crafts and Sewing, and
Video Games. Following prior work~\cite{hou2022towards}, we apply 5-core
filtering, sort each user's interactions chronologically, and truncate each
interaction sequence to a maximum length of 20. Table~\ref{tab:dataset}
summarizes the resulting datasets.

\subsubsection{Baselines}

We compare \model with two groups of methods. The first group contains
\textbf{ID-based recommenders}: Caser~\cite{tang2018personalized},
HGN~\cite{ma2019hierarchical}, GRU4Rec~\cite{hidasi2015session},
SASRec~\cite{kang2018self}, and
S\textsuperscript{3}-Rec~\cite{zhou2020s3}. The second group contains
\textbf{generative recommenders} with different decoding
strategies: P5-CID~\cite{hua2023index,geng2022recommendation},
TIGER~\cite{rajput2023recommender},
LC-Rec~\cite{zheng2024adapting}, RPG~\cite{hou2025generating}, and
CoFiRec~\cite{wei2025cofirec}. Further detail is provided in Appendix C.

\subsubsection{Evaluation Protocol}

We adopt leave-one-out evaluation~\cite{lee2025sequential,kang2018self}. For
each user, the last interaction is used for testing, the penultimate
interaction is used for validation, and all preceding interactions are used
for training. We report HR@\(K\) for \(K\in\{1,5,10\}\) and NDCG@\(K\) for
\(K\in\{5,10\}\). All scoring-based methods, including \model, are evaluated
by ranking the full candidate catalog. Autoregressive baselines use
prefix-constrained beam search with a beam size of 20.

\subsubsection{Implementation Details}

\begin{table}[t]
    \centering
    \caption{Statistics of the three Amazon product categories. Avg. length
    denotes the average number of interactions per user after preprocessing.}
    \label{tab:dataset}
    \begin{adjustbox}{max width=\linewidth}
        \begin{tabular}{@{}lccccc@{}}
            \toprule
            Dataset & \#Users & \#Items & \#Inter. & Sparsity & Avg. length \\
            \midrule
            Instruments & 24,773 & 9,922  & 206,153 & 99.92\% & 8.32 \\
            Arts        & 45,142 & 20,956 & 390,832 & 99.96\% & 8.66 \\
            Games       & 50,547 & 16,859 & 452,989 & 99.95\% & 8.96 \\
            \bottomrule
        \end{tabular}
    \end{adjustbox}
\end{table}

\begin{table*}[t]
    \centering
    \caption{Overall recommendation performance on three Amazon product
    categories. \model achieves the best result in all 15
    dataset--metric settings. Bold and underlined values denote the best and
    second-best results, respectively.}
    \label{tab:mainres}
    \small
    \setlength{\tabcolsep}{3pt}
    \renewcommand{\arraystretch}{1.08}
    \begin{tabularx}{\textwidth}{@{}c|l|*{5}{Y}|*{5}{Y}|Y@{}}
        \toprule
        Dataset & Metric
        & Caser
        & HGN
        & GRU4Rec
        & SASRec
        & S\textsuperscript{3}-Rec
        & P5-CID
        & TIGER
        & LC-Rec
        & CoFiRec
        & RPG
        & \model \\
        \midrule

        \multirow{5}{*}{Instruments}
        & HR@1
        & 0.0182 & 0.0537 & 0.0426 & 0.0209 & 0.0262
        & 0.0587 & \underline{0.0608} & 0.0601 & 0.0592 & 0.0562
        & \textbf{0.0641} \\

        & HR@5
        & 0.0583 & 0.0583 & 0.0754 & 0.0808 & 0.0808
        & 0.0827 & 0.0863 & 0.0844 & \underline{0.0876} & 0.0756
        & \textbf{0.0953} \\

        & HR@10
        & 0.0750 & 0.0750 & 0.0966 & 0.1065 & 0.1079
        & 0.1016 & 0.1064 & \underline{0.1081} & 0.1029 & 0.0902
        & \textbf{0.1156} \\

        & NDCG@5
        & 0.0392 & 0.0392 & 0.0596 & 0.0526 & 0.0565
        & 0.0546 & \underline{0.0738} & 0.0726 & 0.0735 & 0.0619
        & \textbf{0.0802} \\

        & NDCG@10
        & 0.0446 & 0.0446 & 0.0665 & 0.0609 & 0.0643
        & 0.0625 & \underline{0.0803} & 0.0784 & 0.0801 & 0.0666
        & \textbf{0.0868} \\
        \midrule

        \multirow{5}{*}{Arts}
        & HR@1
        & 0.0114 & 0.0355 & 0.0289 & 0.0146 & 0.0233
        & \underline{0.0475} & 0.0465 & 0.0338 & 0.0457 & 0.0459
        & \textbf{0.0488} \\

        & HR@5
        & 0.0379 & 0.0654 & 0.0600 & 0.0752 & 0.0685
        & 0.0724 & \underline{0.0788} & 0.0556 & 0.0759 & 0.0684
        & \textbf{0.0809} \\

        & HR@10
        & 0.0526 & 0.0898 & 0.0805 & 0.1007 & 0.0919
        & 0.0902 & \underline{0.1012} & 0.0729 & 0.0985 & 0.0880
        & \textbf{0.1020} \\

        & NDCG@5
        & 0.0243 & 0.0505 & 0.0450 & 0.0473 & 0.0472
        & 0.0607 & \underline{0.0631} & 0.0450 & 0.0610 & 0.0544
        & \textbf{0.0656} \\

        & NDCG@10
        & 0.0294 & 0.0584 & 0.0516 & 0.0555 & 0.0547
        & 0.0664 & \underline{0.0703} & 0.0505 & 0.0683 & 0.0608
        & \textbf{0.0724} \\
        \midrule

        \multirow{5}{*}{Games}
        & HR@1
        & 0.0086 & 0.0166 & 0.0161 & 0.0145 & 0.0131
        & 0.0177 & 0.0181 & 0.0093 & \underline{0.0192} & 0.0183
        & \textbf{0.0213} \\

        & HR@5
        & 0.0371 & 0.0520 & 0.0548 & 0.0581 & 0.0591
        & 0.0506 & \underline{0.0599} & 0.0300 & 0.0579 & 0.0466
        & \textbf{0.0612} \\

        & HR@10
        & 0.0615 & 0.0837 & 0.0888 & 0.0940 & \underline{0.0963}
        & 0.0803 & 0.0939 & 0.0484 & 0.0903 & 0.0673
        & \textbf{0.0971} \\

        & NDCG@5
        & 0.0228 & 0.0340 & 0.0354 & 0.0365 & 0.0361
        & 0.0342 & \underline{0.0392} & 0.0196 & 0.0381 & 0.0326
        & \textbf{0.0403} \\

        & NDCG@10
        & 0.0306 & 0.0442 & 0.0463 & 0.0481 & 0.0481
        & 0.0437 & \underline{0.0501} & 0.0256 & 0.0485 & 0.0393
        & \textbf{0.0526} \\
        \bottomrule
    \end{tabularx}
\end{table*}

We use Qwen3-4B~\cite{qwen3technicalreport} as the default LLM backbone of
\model. We optimize the model with AdamW~\cite{loshchilov2017decoupled}, a
weight decay of 0.01, and a cosine learning-rate schedule with a warmup ratio
of 0.01. The peak learning rate and the number of training epochs are selected
on the validation set through grid search. Data parallelism and gradient
accumulation yield an effective batch size of 32.

ID-based baselines are implemented with RecBole~\cite{zhao2021recbole}.
Generative baselines follow their official architectures and training
procedures whenever available. We construct the SIDs used by \model with
RQ-VAE, while retaining the baseline-specific SID construction procedures
from their original implementations. LC-Rec uses the same Qwen3-4B backbone
and LoRA~\cite{hu2022lora} settings as \model.

% ----------------------------------------------------------------------
% REQUIRED CONTROLLED STUDY FOR THEORY--EXPERIMENT ALIGNMENT
%
% Insert this subsection after the experimental setup once the results are
% available. Do not infer or fabricate the missing values.
%
% \subsection{Ranking Fidelity under Token-Mediated Inference}
%
% Use the same context encoder, SID tokenizer, training data, and model
% capacity for all variants. Change only the inference interface:
% (1) exhaustive scoring of all legal autoregressive SID sequences;
% (2) beam search with B in {5, 10, 20, 50, 100};
% (3) product of position-wise marginals; and
% (4) direct item scoring.
%
% Report top-k overlap with exhaustive scoring, rank correlation or inversion
% rate, candidate coverage, NDCG, latency, and peak memory. This controlled
% study should provide the direct empirical evidence for the ranking-fidelity
% claims developed in the motivation section.
% ----------------------------------------------------------------------

\subsection{Overall Performance}

Table~\ref{tab:mainres} reports the overall results for RQ1. To complement
this comparison, Appendix~D provides controlled experiments on different SID
constructors, item representations, and decoding strategies under the same
LLM backbone.

\model achieves the best performance in all 15 dataset--metric combinations,
outperforming the strongest baseline in each setting by 4.85\% on average.
Its advantage is consistent across all the datasets and evaluation dimensions:
the average improvement reaches 7.58\% on Instruments, while all six
NDCG comparisons show gains of 2.81\%--8.67\%. These results indicate
that \model improves not only recommendation coverage but also the quality
of the item rankings.

The consistent superiority of \model demonstrates its end-to-end
effectiveness, even against highly competitive generative recommenders.
Specifically, generative methods constitute the strongest baseline in 14 of
the 15 settings, with S\textsuperscript{3}-Rec providing the only exception
on HR@10 for Games. The ability of \model to consistently outperform these
baselines indicates that recommendation performance depends not only on the
generative paradigm itself, but also on how item preferences are modeled and
rankings are produced. 

\subsection{Ranking Agreement and Inference Trade-offs}

\begin{table}
    \centering
    \caption{Average top-10 recommendation overlap under varying beam sizes. We perform the experiments on the "Instruments" dataset. "Inf.time" denotes the average inference time for one user on a single NVIDIA A100 GPU.}
    \small
    \begin{tabular}{l|ccc}
    \toprule
        Generation Method & Overlap & NDCG@10 & Inf.time (s) \\
         \midrule
         Beam size 50 & 90.28\% & 0.0802 & 0.683 \\
         Beam size 100 & 94.76\% & 0.0816 & 0.981 \\
         Parallel & 86.32\% & 0.0791 & 0.082 \\
         \model & 91.20\% & 0.0868 & 0.074 \\
         Exhaustive scoring & 100.00\% & 0.0842 & 5.052 \\
         \bottomrule
    \end{tabular}
    
    \label{tab:rank}
\end{table}

Because complete user-specific preference rankings are unavailable in
real-world recommendation datasets, we use the exhaustive full-catalog
ranking of LC-Rec as a proxy. Specifically, we score every legal SID sequence to
recover the exact ranking induced by the trained LC-Rec model, thereby
eliminating the approximation introduced by decoding. We then compare this
reference with the rankings produced by beam search, parallel generation,
and \model in terms of top-$K$ overlap, recommendation accuracy, and
inference efficiency. The results for RQ2 are reported in
Table~\ref{tab:rank}.

The results reveal distinct limitations of the two existing generation
paradigms. For autoregressive generation, increasing the beam size improves
agreement with the exhaustive LC-Rec ranking, raising the top-10 overlap
from 90.28\% to 94.76\%, but also increases the per-user inference time by
more than 40\%. This demonstrates that more faithful recovery of the
model-induced ranking requires increasingly expensive search. Parallel
generation largely removes this computational burden, yet produces the
lowest overlap and NDCG@10 among the evaluated methods. Its efficiency
therefore comes with more substantial ranking distortion, consistent with
the information loss caused by position-wise SID prediction.

In contrast, \model achieves a top-10 overlap of 91.20\%, surpassing both
beam search with a beam size of 50 and parallel generation. This agreement
should not be interpreted as the primary measure of recommendation quality,
because the reference reflects the ranking learned by LC-Rec rather than the
latent ground-truth preference order. More importantly, \model obtains the
highest NDCG@10 of 0.0868, exceeding exhaustive LC-Rec scoring by 3.09\%,
while requiring only 0.074\,s per user. It is therefore over nine times
faster than the evaluated beam-search configurations and slightly faster
than parallel generation. These results show that direct item-level scoring
can learn a more accurate recommendation ranking and recover its own
model-induced order efficiently, without relying on token-level decoding or
its associated accuracy--efficiency trade-off.

\subsection{Backbone Analysis}

\begin{figure}[t]
    \centering
    \includegraphics[width=\linewidth]{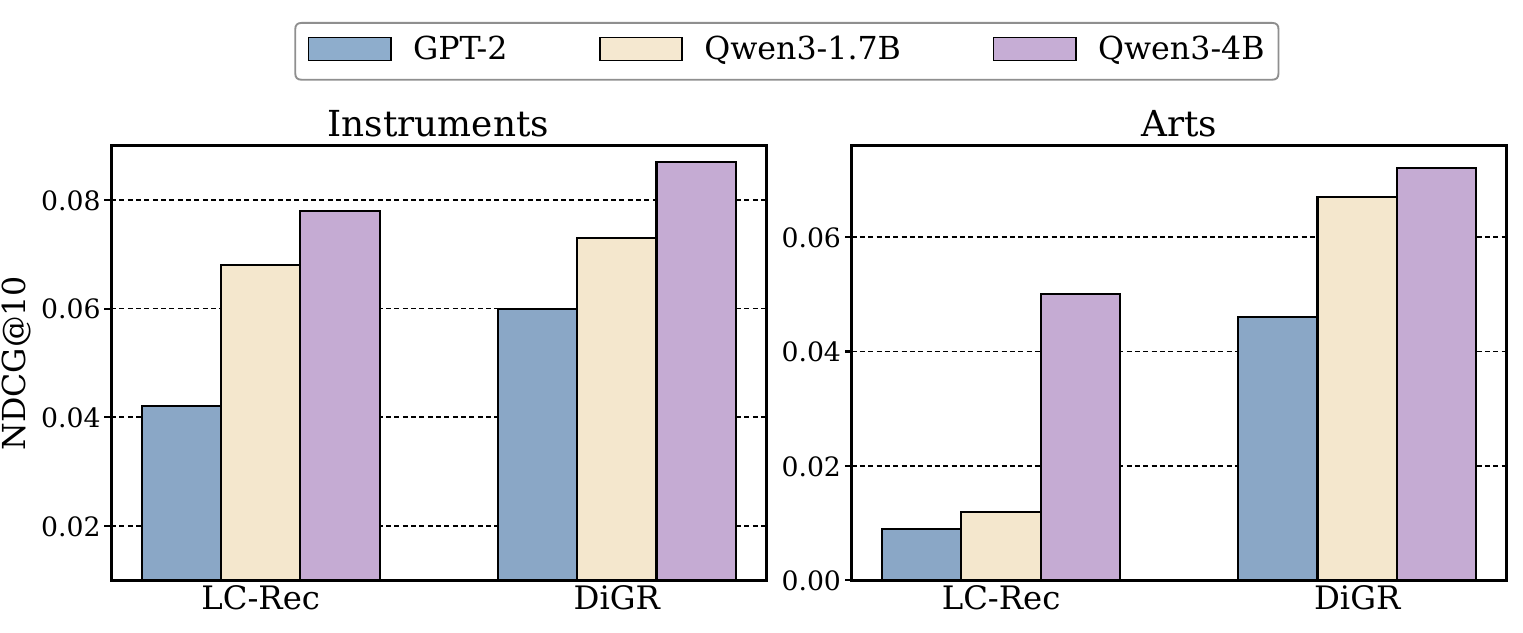}
    \caption{\model achieves higher NDCG@10 than LC-Rec under every evaluated
    backbone configuration on Instruments and Arts.}
    \label{fig:scale}
\end{figure}

To answer RQ3, we evaluate \model and LC-Rec with GPT-2
~\cite{radford2019language}, Qwen3-1.7B, and Qwen3-4B on Instruments and
Arts. For each backbone, the two methods adopt the same model configuration
and fine-tuning strategy: GPT-2 is fully fine-tuned, whereas LoRA is applied
to the Qwen3 models. Figure~\ref{fig:scale} reports the corresponding
NDCG@10 results.

Performance generally improves across the backbone scales,
demonstrating that both methods benefit from increased model capacity.
Nevertheless, \model consistently outperforms LC-Rec on both datasets and
under every backbone configuration. Its advantage is pronounced with the
smaller backbones, suggesting that its item-level formulation makes more
effective use of limited model capacity. As the backbone scales to Qwen3-4B,
\model continues to improve while preserving its performance advantage.
These results show that \model performs strongly without requiring a large
backbone, while remaining capable of exploiting the additional capacity
provided by larger LLMs.

\subsection{Effect of SID Composition}

\begin{table}[t]
    \centering
    \caption{Effect of different aggregators on Instruments. Parameter-free mean
    pooling achieves the best result on all five metrics. Bold and underlined
    values denote the best and second-best results, respectively.}
    \label{tab:agg}
    \small
    \begin{adjustbox}{width=0.47\textwidth}
        \begin{tabular}{l|ccccc}
            \toprule
            Aggregator & Concat. & Proj. & Attention & Gating & Mean \\
            \midrule
            HR@1    & 0.0601 & 0.0587 & 0.0606 & \underline{0.0621} & \textbf{0.0641} \\
            HR@5    & 0.0888 & 0.0867 & 0.0889 & \underline{0.0896} & \textbf{0.0953} \\
            HR@10   & \underline{0.1114} & 0.1051 & 0.1102 & 0.1086 & \textbf{0.1156} \\
            NDCG@5  & 0.0748 & 0.0730 & 0.0754 & \underline{0.0763} & \textbf{0.0802} \\
            NDCG@10 & 0.0821 & 0.0789 & 0.0822 & \underline{0.0825} & \textbf{0.0868} \\
            \bottomrule
        \end{tabular}
    \end{adjustbox}
\end{table}

In \model, the SID token embeddings are aggregated into a
single item representation. To answer RQ4, we compare the default mean
pooling operator with concatenation-based aggregation, projection followed
by pooling, gating, and cross-attention, while holding all other components
fixed. Table~\ref{tab:agg} reports the results on Instruments.

The results show that a simple composition operator is sufficient for
constructing effective item representations in \model. Mean pooling achieves
the best performance across all five metrics, while gating ranks second on
four, and none of the more expressive operators provides a further
improvement. This suggests that the shared SID-token embeddings already
provide an informative basis for item composition, reducing the need for an
additional learnable aggregation mechanism. By directly averaging these
embeddings, mean pooling preserves their shared semantic structure while
avoiding the additional parameters and optimization burden introduced by
more complex operators.

\subsection{Exposure Diversity}

\begin{figure}[t]
    \centering
    \includegraphics[width=0.8\linewidth]{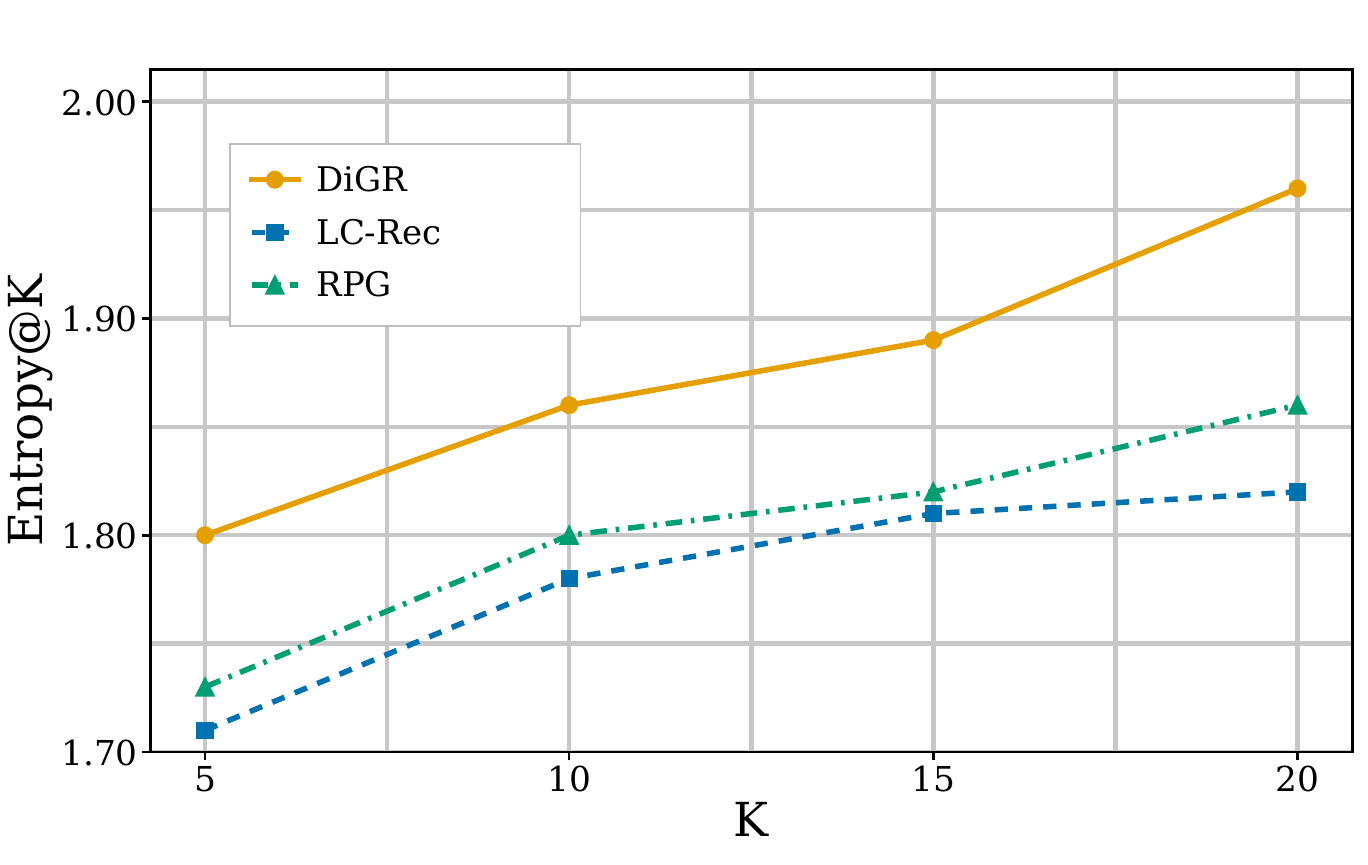}
    \caption{\model achieves the highest category-level Entropy@\(K\) at every
    evaluated cutoff on Instruments.}
    \label{fig:entropy}
\end{figure}

Beyond ranking accuracy, RQ5 examines the diversity of the
recommendation lists. Following prior work~\cite{rajput2023recommender}, we
adopt category-level Entropy@$K$, which measures the entropy of the category
distribution among the top-$K$ recommended items. A higher value represents
a broader and more balanced distribution across categories.
Figure~\ref{fig:entropy} compares \model with LC-Rec and RPG on Instruments.

\model achieves the highest Entropy@$K$ at every evaluated cutoff,
demonstrating a consistent advantage in recommendation diversity. At
$K=20$, its entropy reaches 1.962, compared with 1.859 for RPG and 1.816 for
LC-Rec. This indicates that \model is less likely to concentrate
on a small number of categories and instead produces
more varied recommendation lists. Together with its accuracy improvements in
RQ1, these results show that \model achieves greater diversity without
sacrificing ranking quality.

\section{Conclusions}
From the perspective of item-ranking preservation, position-wise parallel
marginals are generally ranking-insufficient and preclude exact recovery,
whereas autoregressive generation is ranking-sufficient: exact decoding is
output-equivalent to item-level scoring, but practical beam search may
introduce pruning errors. This motivates an item-atomic SID-based model
family and DiGR, which composes shared SID-token embeddings into item
representations and directly optimizes inference scores. Experiments on three
datasets and multiple LLM backbones show consistent gains in accuracy,
ranking quality, and exposure diversity.
% \appendix
% \input{content/appendix}

\bibliography{aaai2027}

\appendix
\section{A. Detailed Proof to Proposition 1}

\begin{proof}
Let $\mathcal{P}_{\mathcal{G}}$ denote the set of model-induced item
distributions, and let $\mathfrak{S}_{\mathcal{I}}$ denote the set of total
rankings over $\mathcal{I}$ under a fixed tie-breaking rule. Formally, the
two statements in the proposition are:

\begin{equation}
\tag{1}
\forall p,\widetilde{p}\in\mathcal{P}_{\mathcal{G}},
\quad
\Phi_{\mathcal{G}}(p)
=
\Phi_{\mathcal{G}}(\widetilde{p})
\ \Longrightarrow\
R(p)=R(\widetilde{p}),
\end{equation}
which is the definition of global ranking sufficiency; and

\begin{equation}
\tag{2}
\begin{aligned}
\exists\,
\mathcal{S}_{\mathcal{G}}:
\operatorname{Range}(\Phi_{\mathcal{G}})
&\rightarrow
\mathbb{R}^{|\mathcal{I}|}
\quad\text{such that}\\
\forall p\in\mathcal{P}_{\mathcal{G}},
\quad
\operatorname{argsort}^{\downarrow}_{i\in\mathcal{I}}
\mathcal{S}_{\mathcal{G}}
\bigl(\Phi_{\mathcal{G}}(p)\bigr)_i
&=
R(p).
\end{aligned}
\end{equation}

We first prove $(1)\Rightarrow(2)$. Suppose that $(1)$ holds. Define
\[
\overline{R}:
\operatorname{Range}(\Phi_{\mathcal{G}})
\rightarrow
\mathfrak{S}_{\mathcal{I}}
\]
by
\[
\overline{R}(y)=R(p)
\quad\text{for any }p\in\mathcal{P}_{\mathcal{G}}
\text{ satisfying }\Phi_{\mathcal{G}}(p)=y.
\]
Such a distribution exists because
$y\in\operatorname{Range}(\Phi_{\mathcal{G}})$. Moreover, this definition
is independent of the choice of $p$. Indeed, if
\[
\Phi_{\mathcal{G}}(p)
=
\Phi_{\mathcal{G}}(\widetilde{p})
=
y,
\]
then $(1)$ implies that $R(p)=R(\widetilde{p})$. Hence,
$\overline{R}$ is well-defined.

Let $n=|\mathcal{I}|$. For a ranking
$r=(i_1,\ldots,i_n)\in\mathfrak{S}_{\mathcal{I}}$, define its score
encoding $E(r)\in\mathbb{R}^{n}$ by
\[
E(r)_{i_k}=n-k,
\qquad k=1,\ldots,n.
\]
The coordinates of $E(r)$ are distinct and decrease according to $r$.
Therefore,
\[
\operatorname{argsort}^{\downarrow}_{i\in\mathcal{I}}E(r)_i=r.
\]
Now define
\[
\mathcal{S}_{\mathcal{G}}
=
E\circ\overline{R}.
\]
Then, for every $p\in\mathcal{P}_{\mathcal{G}}$,
\begin{align}
\operatorname{argsort}^{\downarrow}_{i\in\mathcal{I}}
\mathcal{S}_{\mathcal{G}}
\bigl(\Phi_{\mathcal{G}}(p)\bigr)_i
&=
\operatorname{argsort}^{\downarrow}_{i\in\mathcal{I}}
E\bigl(\overline{R}(\Phi_{\mathcal{G}}(p))\bigr)_i\\
&=
\overline{R}(\Phi_{\mathcal{G}}(p))\\
&=
R(p).
\end{align}
Thus, $(2)$ holds.

We next prove $(2)\Rightarrow(1)$. Suppose that $(2)$ holds, and take any
$p,\widetilde{p}\in\mathcal{P}_{\mathcal{G}}$ satisfying
\[
\Phi_{\mathcal{G}}(p)
=
\Phi_{\mathcal{G}}(\widetilde{p}).
\]
Because $\mathcal{S}_{\mathcal{G}}$ depends only on the interface output,
\[
\mathcal{S}_{\mathcal{G}}
\bigl(\Phi_{\mathcal{G}}(p)\bigr)
=
\mathcal{S}_{\mathcal{G}}
\bigl(\Phi_{\mathcal{G}}(\widetilde{p})\bigr).
\]
Taking the descending argsort of both sides and applying $(2)$ gives
\[
R(p)=R(\widetilde{p}).
\]
Therefore, $(1)$ holds, and $\Phi_{\mathcal{G}}$ is globally
ranking-sufficient.

Finally, if the equivalent statements fail, the negation of $(1)$ gives
$p,\widetilde{p}\in\mathcal{P}_{\mathcal{G}}$ such that
\[
\Phi_{\mathcal{G}}(p)
=
\Phi_{\mathcal{G}}(\widetilde{p}),
\qquad
R(p)\neq R(\widetilde{p}).
\]
Any decoder depending only on $\Phi_{\mathcal{G}}$ must receive the same
input and return the same output for these two distributions. It therefore
cannot recover both rankings exactly. The same argument applies to any
item-level score mapping based only on $\Phi_{\mathcal{G}}$.
\end{proof}

\section{B. Detailed Proof to Theorem 1}

\begin{proof}
Let $n=|\mathcal{I}|$ and $m_l=|\mathcal{V}_l|$. For each item
$i\in\mathcal{I}$, let
\[
\mathbf{s}_i=(s_{i,1},\ldots,s_{i,L})
\]
denote its SID. For each position $l$, define the incidence matrix
\[
A_l\in\{0,1\}^{m_l\times n}
\]
by
\[
(A_l)_{v,i}
=
\mathbb{I}[s_{i,l}=v],
\qquad
v\in\mathcal{V}_l,\ i\in\mathcal{I}.
\]
Thus, for any item distribution
$p\in\Delta^{n-1}$, its marginal distribution at position $l$ is
\[
A_l p,
\]
and the parallel token interface can be written as the linear mapping
\[
\Phi_{\mathrm{PAR}}(p)
=
Ap,
\qquad
A
=
\begin{bmatrix}
A_1\\
\vdots\\
A_L
\end{bmatrix}.
\]

Because every item has exactly one token at each SID position, for every
$x\in\mathbb{R}^{n}$ and every $l$,
\[
\mathbf{1}_{m_l}^{\top}A_lx
=
\mathbf{1}_{n}^{\top}x.
\]
Consequently, the image of $A$ is contained in the linear subspace
\[
\mathcal{W}
=
\left\{
(y_1,\ldots,y_L):
\mathbf{1}_{m_1}^{\top}y_1
=
\cdots
=
\mathbf{1}_{m_L}^{\top}y_L
\right\}.
\]
The $L-1$ equalities defining $\mathcal{W}$ are linearly independent.
Therefore,
\begin{align}
\dim(\mathcal{W})
&=
\sum_{l=1}^{L}m_l-(L-1)\\
&=
1+\sum_{l=1}^{L}(m_l-1).
\end{align}
It follows that
\[
\operatorname{rank}(A)
\leq
1+\sum_{l=1}^{L}(m_l-1).
\]
Under the condition of the theorem,
\[
n
>
1+\sum_{l=1}^{L}(m_l-1),
\]
and hence
\[
\operatorname{rank}(A)<n.
\]
By the rank--nullity theorem, there exists a nonzero vector
$h\in\mathbb{R}^{n}$ such that
\[
Ah=0.
\]

Moreover, for any position $l$,
\[
\mathbf{1}_{n}^{\top}h
=
\mathbf{1}_{m_l}^{\top}A_lh
=
0.
\]
Since $h\neq 0$ and its coordinates sum to zero, $h$ must contain at least
one positive and one negative coordinate. Hence, there exist
$a,b\in\mathcal{I}$ such that
\[
h_a>0,
\qquad
h_b<0.
\]

Let $u\in\Delta^{n-1}$ be the uniform distribution,
\[
u_i=\frac{1}{n},
\qquad i\in\mathcal{I}.
\]
Choose
\[
0<\epsilon<
\frac{1}{n\lVert h\rVert_{\infty}}
\]
and define
\[
p^{+}=u+\epsilon h,
\qquad
p^{-}=u-\epsilon h.
\]
The choice of $\epsilon$ ensures that every coordinate of $p^{+}$ and
$p^{-}$ is strictly positive. Furthermore,
\[
\mathbf{1}_{n}^{\top}p^{+}
=
\mathbf{1}_{n}^{\top}p^{-}
=
1,
\]
because $\mathbf{1}_{n}^{\top}h=0$. Therefore,
\[
p^{+},p^{-}\in\Delta^{n-1}.
\]

Since $Ah=0$, the two distributions induce exactly the same position-wise
token marginals:
\begin{align}
\Phi_{\mathrm{PAR}}(p^{+})
&=
A(u+\epsilon h)
=
Au+\epsilon Ah
=
Au,\\
\Phi_{\mathrm{PAR}}(p^{-})
&=
A(u-\epsilon h)
=
Au-\epsilon Ah
=
Au.
\end{align}
Thus,
\[
\Phi_{\mathrm{PAR}}(p^{+})
=
\Phi_{\mathrm{PAR}}(p^{-}).
\]

On the other hand, the relative order of items $a$ and $b$ is reversed:
\[
p^{+}_a-p^{+}_b
=
\epsilon(h_a-h_b)
>
0,
\]
whereas
\[
p^{-}_a-p^{-}_b
=
-\epsilon(h_a-h_b)
<
0.
\]
Therefore, item $a$ is ranked above item $b$ under $p^{+}$ but below item
$b$ under $p^{-}$. Hence,
\[
\mathcal{R}(p^{+})
\neq
\mathcal{R}(p^{-}).
\]
Because $p^{+}$ and $p^{-}$ are strictly positive, they are realizable
under the assumed model family and are therefore valid model-induced
preference distributions.

Finally, let
\[
D:
\operatorname{Range}(\Phi_{\mathrm{PAR}})
\rightarrow
\mathfrak{S}_{\mathcal{I}}
\]
be any decoder based only on the position-wise marginals. Since the two
interface outputs are identical,
\[
D\bigl(\Phi_{\mathrm{PAR}}(p^{+})\bigr)
=
D\bigl(\Phi_{\mathrm{PAR}}(p^{-})\bigr).
\]
However, their exact model rankings are different. Thus, the common decoder
output cannot equal both $\mathcal{R}(p^{+})$ and
$\mathcal{R}(p^{-})$. No such decoder can therefore recover the exact
ranking for every model-induced preference distribution, and
$\Phi_{\mathrm{PAR}}$ is not globally ranking-sufficient.
\end{proof}

\section{C. Baselines}
\begin{table*}[t]
\centering
\caption{Controlled comparison across different item representations. The best results are shown in \textbf{bold}.}
\label{tab:controlled_comparison}
\small
\setlength{\tabcolsep}{8pt}
\begin{tabular}{lcccccc}
\toprule
& \multicolumn{2}{c}{AR}
& \multicolumn{2}{c}{Parallel}
& \multicolumn{2}{c}{\model} \\
\cmidrule(lr){2-3}
\cmidrule(lr){4-5}
\cmidrule(lr){6-7}
Representation
& HR@10 
& NDCG@10 
& HR@10 
& NDCG@10 
& HR@10 
& NDCG@10  \\
\midrule
RQ-VAE
& 0.0709 & 0.0423
& 0.0713 & 0.0496
& \textbf{0.0916} & \textbf{0.0602} \\

R-Kmeans
& 0.0705 & 0.0437
& 0.0682 & 0.0407
& \textbf{0.0809} & \textbf{0.0575} \\

OPQ
& 0.0610 & 0.0392
& 0.0721 & 0.0501
& \textbf{0.0830} & \textbf{0.0611} \\

Title
& 0.0033 & 0.0015
& N/A & N/A
& \textbf{0.0084} & \textbf{0.0043} \\

ID
& 0.0335 & 0.0174
& N/A & N/A
& \textbf{0.0827} & \textbf{0.0482} \\
\bottomrule
\end{tabular}
\end{table*}
We adopt the following representative sequential recommendation models as baselines for comparison with our proposed framework \model: 

\begin{itemize}
    \item \textbf{Caser}~\cite{tang2018personalized} leverages convolutional neural networks to extract local and global sequential patterns by horizontal and vertical convolutional filters.
    \item \textbf{HGN}~\cite{ma2019hierarchical} introduces a hierarchical gating mechanism that dynamically balances short-range and long-range user preferences in sequential recommendation.
    \item \textbf{GRU4Rec}~\cite{hidasi2015session} models user interaction sequences with recurrent neural networks, capturing temporal dynamics through gated recurrent units.
    \item \textbf{SASRec}~\cite{kang2018self} adopts a causal self-attention mechanism to model unidirectional sequential dependencies of user interaction sequence for next-item recommendation.
    \item \textbf{S\textsuperscript{3}-Rec}~\cite{zhou2020s3} pre-train a self-supervised sequential recommendation model via maximizing mutual information to learn the correlation between items and attributes.
    \item \textbf{P5-CID}~\cite{hua2023index,geng2022recommendation} reformulates multiple recommendation objectives into a unified text-to-text paradigm and jointly models these tasks using the T5 architecture~\cite{ni2022sentence}.
    \item \textbf{TIGER}~\cite{rajput2023recommender} is the first to apply generative retrieval paradigm on sequential recommendation task and to introduce semantic IDs for item representation.
    \item \textbf{LC-Rec}~\cite{zheng2024adapting} is a representative autoregressive generative recommender. It also designs auxiliary fine-tuning tasks to align language semantics with collaborative semantics.
    \item \textbf{CoFiRec}~\cite{wei2025cofirec} represents items with semantic IDs and collaborative IDs simultaneously. It also introduces additional learnable position embeddings to better capture structural dependencies.
    \item \textbf{RPG}~\cite{hou2025generating} utilizes long semantic IDs and generate the tokens under the parallel generation paradigm, achieving great recommendation performance and short inference latency.
\end{itemize}

\section{D. Extensive Experimental Study}

We provide experiments of different SID constructors, item representations and decoding strategies with the same LLM backbone. For SID constructors, we use RQ-VAE~\cite{rajput2023recommender}, R-KMeans~\cite{deng2025onerec} and OPQ~\cite{hou2025generating}. For item representations, we use SID, item title and vanilla ID. For decoding strategies, we use autoregressive, parallel and \model for comparison. For brevity, we use GPT-2 as the LLM backbone and dot product for parallel generation. The experiments are performed on the Instruments dataset. Table~\ref{tab:controlled_comparison} exhibits the results. Note that parallel decoding for ID is a single step generation exactly same as \model, while title has unfixed length and is not applicable to parallel decoding.

\model consistently delivers the strongest recommendation performance
across the evaluated item representations. It achieves the best HR@10
and NDCG@10 in every applicable comparison, outperforming both AR
and Parallel under all SID constructions. This result supports the
effectiveness of directly optimizing item-level rankings instead of
recovering them through token predictions.

\model is also more robust to the choice of SID constructor. The relative
performance of AR and Parallel changes across RQ-VAE, R-Kmeans, and
OPQ, whereas \model remains the best-performing paradigm in every case.
This consistency suggests that item-level scoring provides a more stable
interface for exploiting different SID structures.

Semantic SID representations generally provide more effective item
representations than titles or atomic IDs. RQ-VAE produces the highest
hit rate, while OPQ achieves the highest ranking quality; by contrast,
the title-based variant performs substantially worse, and the atomic-ID
variant remains weaker in NDCG@10. This pattern indicates that the
compositional information encoded by SIDs contributes meaningfully to
the ranking performance of \model.

\end{document}